\documentclass[conference]{IEEEtran}
\IEEEoverridecommandlockouts
\usepackage{cite}
\usepackage{amsmath,amssymb,amsfonts}
\usepackage{algorithmic}
\usepackage{graphicx}
\usepackage{textcomp}
\usepackage{xcolor}
\usepackage{commath}

\usepackage{verbatim}       
\usepackage{amsmath}
\usepackage{amsthm}
\usepackage{amssymb}
\usepackage{mathtools}
\usepackage{mathrsfs}
\usepackage{algorithmic}
\usepackage{algorithm}
\usepackage[caption=false,font=footnotesize]{subfig}

\usepackage{microtype}

\newtheorem{lemma}{Lemma}
\newtheorem{corollary}{Corollary}
\newtheorem{proposition}{Proposition}

\newtheorem{assumption}{Assumption}

\newcommand{\E}{\mathbb{E}}

\newcommand{\diag}{\ensuremath{\operatorname{diag}}}

\def\BibTeX{{\rm B\kern-.05em{\sc i\kern-.025em b}\kern-.08em
    T\kern-.1667em\lower.7ex\hbox{E}\kern-.125emX}}
\begin{document}

\title{Numerically robust square root implementations of statistical linear regression filters and smoothers
\thanks{Filip Tronarp was partially supported by the Wallenberg AI, Autonomous
Systems and Software Program (WASP) funded by the Knut and Alice Wallenberg Foundation.}
}

\author{\IEEEauthorblockN{Filip Tronarp}
\IEEEauthorblockA{\textit{Centre for Mathematical Sciences} \\
\textit{Lund University}\\
Lund, Sweden \\
filip.tronarp@matstat.lu.se}
}

\maketitle

\begin{abstract}
In this article, square-root formulations of
the statistical linear regression filter and smoother are developed.
Crucially, the method uses QR decompositions rather than Cholesky downdates.
This makes the method inherently more numerically robust than the downdate based methods,
which may fail in the face of rounding errors.
This increased robustness is demonstrated in an ill-conditioned problem,
where it is compared against a reference implementation in both double and single precision arithmetic.
The new implementation is found to be more robust, when implemented in lower precision arithmetic as compared to the alternative.
\end{abstract}


\section{Introduction}
Consider the probabilistic state-space model given by
\begin{subequations}
\begin{align}
x_0 &\sim \mathcal{N}(\mu_0, \Sigma_0), \\
x_m \mid x_{m-1} &\sim \mathcal{N}(f(x_{m-1}) , Q(x_{m-1}) ), \\
y_m \mid x_m &\sim \mathcal{N}( c(x_m), R(x_m)),
\end{align}
\end{subequations}
where $f$ is the transition function, $Q$ the transition covariance, $c$ the observation function, $R$ the observation covariance.
The vector $x_m$ is the latent and $y_m$ is an imperfect observation of the former.
Let $\pi_{m\mid 1:k}$ denote the following conditional density
\begin{equation*}
\pi_{m\mid 1:k}(x) = \pi(x_m \mid y_{1:k})
\end{equation*}
so that the \emph{filtering density} at time $m$ is given by $\pi_{m\mid 1:m}$ and the smoothing
density based on measurements $y_1, \ldots, y_n$ is given by $\pi_{m\mid 1:n}$ for $m  \leq n$.

The filtering density, $\pi_{m\mid 1:m}$, satisfies the following prediction/correction recursion \cite{Ho1964,Sarkka2023}
\begin{subequations}
\begin{align*}
\pi_{m+1\mid 1:m}(x) = \int \mathcal{N}(x; f(x'), Q(x')) \pi_{m \mid 1:m}(x') \dif x', \\
\pi_{m+1 \mid 1:m+1}(x) = \frac{\mathcal{N}(y; c(x), R(x)) \pi_{m+1\mid 1:m}(x) }{ \int  \mathcal{N}(y; c(x'), R(x')) \pi_{m+1\mid 1:m}(x') \dif x'}.
\end{align*}
\end{subequations}
Now define backward transition density $b_{m\mid m+1}$:
\begin{equation}\label{eq:backward-kernel}
b_{m\mid m+1}(x \mid x') = \frac{ \pi_{m \mid 1:m}(x) \mathcal{N}(x'; f(x), Q(x))}{\pi_{m + 1 \mid 1:m}(x')},
\end{equation}
then the smoothing density follows the following, backward, recursion \cite{Kitagawa1987,Cappe2009}
\footnote{
In linear Gaussian state space models,
this is equivalent to the Rauch--Tung--Striebel recursion,
see for example \cite{Tronarp2022}.
}
\begin{equation*}
\pi_{m \mid 1:n}(x) = \int b_{m \mid m+1}(x\mid x') \pi_{m+1 \mid 1:n}(x') \dif x'.
\end{equation*}
It may be noted that $b_{m\mid m+1}(x \mid x')$ is simply Bayes' rule when the prior is $\pi_{m \mid 1:m}(x)$ and the likelihood is $ \mathcal{N}(x'; f(x), Q(x))$.

The problem of Bayesian filtering and smoothing for linear Gaussian models was largely solved by
the development of the Kalman filter and Rauch--Tung--Striebel smoother \cite{Kalman1961,Rauch1965}.
However, the non-linear case is intractable in general,
hence the flourishing field of approximate inference \cite{Sarkka2023}.

Perhaps the earliest approach was based on linearization by Taylor series, leading to the \emph{extended Kalman filter} (EKF) \cite{Kailath2000}.
This thinking lead to the optimization based direction in approximate state estimation \cite{Bell1993,Bell1994,Skoglund2015,Aravkin2017,Gao2019}.

Later on, alternative approaches were devised based on cubature based moment approximation and
linear minimum mean square error estimation \cite{Julier2000,Ito2000,Julier2004,Wu2006,Arasaratnam2009}.
The cubature approach has been shown to be equivalent to an alternative linearization strategy
termed \emph{statistical linear regression} \cite{Lefebvre2002,Arasaratnam2007}.

Taking inspiration from the optimization based viewpoint, iterative variants of the statistical linear regression method have been developed \cite{Garcia2015,Garcia2016,Tronarp2018,Garcia2019a,Tronarp2019a,Skoglund2019}.
Iterative statistical linear regression has also found applications in importance density design \cite{Hostettler2020} and Gaussian process classification \cite{Garcia2019b}.

When the covariance matrices are ill-conditioned,
the Ricatti recursion for the covariance in linear state estimators is prone to numerical instability.
This lead to the development of recursions for the square-root of the covariance matrix instead \cite{Kaminski1971,Kailath2000,Anderson2012}.
The reader may consult \cite{Kailath2000} for a catalog of such approaches.
Because Taylor linearization does not modify the covariance matrix, this trivially extends to the EKF.
Various square root methods have also been employed for the cubature based methods \cite{Arasaratnam2008,Tang2012,Bhaumik2014,Menegaz2015,Arasaratnam2011,Milschewski2017}.
However, analogous methods for statistical linear regression based estimators have not been examined to the same extent.
As far as we are aware, the only contribution in this direction is given by \cite{Yaghoobi2022},
which proposes a method based on sequential updating and downdating.
However, the downdating problem can be ill-conditioned \cite{Stewart1979,Bojanczyk1987,Bojanczyk1991},
and may fail completely due to round off errors \cite{Milschewski2017}.

Other notable solutions to the state estimation problem are the projection based approach \cite{Tronarp2019b},
variational inference \cite{Gultekin2017,Courts2021},
and sequential Monte Carlo \cite{Gordon1993,Arulampalam2002,Djuric2003,Doucet2000,Doucet2001,Cappe2007}.
However, they fall outside of the scope of the present discussion.

\subsection{Contribution}
This article develops an approach for implementing statistical linear regression based state estimators in square root form.
In contrast to \cite{Yaghoobi2022}, the proposed method does not rely on Cholesky downdates, but rather QR decompositions.
Consequently, the proposed algorithm is expected to remain numerically reliable when the state estimation problem is ill-conditioned
or when the state estimation algorithm is implemented in low precision arithmetic,
whereas Cholesky downdates may fail \cite{Milschewski2017}.

\subsection{Notation}
For a matrix $A$, its adjoint is denoted by $A^*$.
For a positive definite matrix $\Pi$, its lower triangular Cholesky factor is denoted by $\Pi^{1/2}$
so that $\Pi = \Pi^{1/2}\Pi^{*/2}$, where $\Pi^{*/2}$ is the upper triangular Cholesky factor.
Furthermore, for matrices $A_1$ and $A_2$ the relation $A_1 \cong A_2$ holds if and only if
$A_1^* A_1 = A_2^* A_2$.
In particular, $A_1 \cong (A_1^* A_1)^{*/2}$, and a Cholesky factor of $A_1^* A_1$ may thus be obtained
by taking the upper triangular factor in the QR decomposition of $A_1$.

\section{Gaussian state estimation in square-root form}
In this section, the square root implementations of Gaussian state estimators is summarized.
That is, filtering and smoothing the following model.

\begin{subequations}
\begin{align}
x_m \mid x_{m-1} &\sim \mathcal{N}(\Phi_m x_{m-1}, Q_m), \\
y_m \mid x_m &\sim \mathcal{N}(C_m x_m, R_m).
\end{align}
\end{subequations}

The following lemma is fundamental for the implementation of square-root based filters and smoothers in linear Gaussian models \cite[Chapter 6]{Anderson2012}.

\begin{lemma}\label{lem:inversion}
The model following model
\begin{align*}
u &\sim \mathcal{N}(\bar{u}, \Pi), \\
v \mid u &\sim \mathcal{N}(\Psi u, \Omega),
\end{align*}
is equivalent to
\begin{align*}
v &\sim \mathcal{N}(C \bar{u}, P), \\
u \mid v &\sim \mathcal{N}(\bar{u} + \Gamma (v - \Psi \bar{u}), \Sigma),
\end{align*}
where the parameters $P, \Gamma$, and $\Sigma$ are given by
\begin{align}
\begin{pmatrix} \Omega^{*/2} & 0 \\ \Pi^{*/2}\Psi^* & \Pi^{*/2} \end{pmatrix}
&\cong \begin{pmatrix} P^{*/2} & \bar{\Gamma}^* \\ 0 & \Sigma^{*/2} \end{pmatrix}, \\
\Gamma &= \bar{\Gamma} P^{-1/2}.
\end{align}
\end{lemma}

The interpretation of lemma \ref{lem:inversion} is as follows.
When $u = x_n$, $\bar{u}$ and $\Pi$ are the respective moments of the predictive distribution,
and $v = y_n$, then lemma \ref{lem:inversion} provides a way to compute the Cholesky factors of the filtering
and marginal measurement covariances, and filter gain by just one QR decomposition and one linear solve with respect to
$P^{1/2}$ (the Cholesky factor of the marginal measurement covariance).

Similarly, when $u = x_{n-1}$, $\bar{u}$ and $\Pi$ are the respective moments of the filtering distribution distribution,
and $v = x_n$, then lemma \ref{lem:inversion} provides a way to compute the Cholesky factors of the predictive
and backward process covariances (the covariance of \eqref{eq:backward-kernel}), and smoother gain by just one QR decomposition and one linear solve with respect to
$P^{1/2}$ (the Cholesky factor of the predictive covariance).

It is important to note that lemma \ref{lem:inversion} only requires one linear solve with the Cholesky factor of $P$,
rather than two which is required in the cubature approaches of \cite{Arasaratnam2009,VanDerMerwe2001}.
Consequently, if statistical linear regression can be used to directly obtain a Cholesky factor of $P$,
then it can be expected to be more numerically robust when $P$ is ill-conditioned.

\section{Statistical linear regression in square root form}
In this section, a method for implementing the statistical linear regression algorithm in square root form is developed.
In contrast to the method of \cite{Yaghoobi2022}, it does not make use of Cholesky downdates but
employs QR decompositions instead (equivalent to Cholesky updates).

\subsection{Statistical linear regression}
The statistical linear regression problem is concerned with finding a linear Gaussian approximation
to the conditional density in the following model \cite{Lefebvre2002}
\begin{align}
u &\sim \mathcal{N}(\bar{u}, \Pi), \\
v \mid u &\sim \mathcal{N}(a(u), \Omega),
\end{align}
by the following minimization problem
\begin{subequations}
\begin{align*}
e_{\Psi_0,b_0}(u, v) &= v - \Psi_0 u - b_0,\\
\bar{\Psi}, \bar{b} &= \arg\,\min_{\Psi_0, b_0} \, \E\Big[\abs{e_{\Psi_0,b_0}(u, v)}^2\Big].
\end{align*}
\end{subequations}
Let $\bar{e}(u,v) = e_{\bar{\Psi}, \bar{b}}(u,v)$, then the residual covariance matrix is given by
\begin{equation}
\bar{\Omega} = \E[\bar{e}(u, v) \bar{e}^*(u, v)].
\end{equation}
This results in the following Gaussian approximation \cite{Garcia2015,Garcia2016}
\begin{equation*}
\mathcal{N}(a(u), \Omega) \approx \mathcal{N}(\bar{\Psi}u + \bar{b}, \bar{\Omega}).
\end{equation*}
The solution to the statistical linear regression problem is given by \cite{Lefebvre2002,Garcia2015,Garcia2016}
\begin{subequations}\label{eq:slrexact}
\begin{align}
\bar{a} &= \E[a(u)], \\
\bar{\Psi} &= \E[a(u)(u - \bar{u})^*] \Pi^{-1}, \\
\bar{b} &= \bar{a} - \bar{\Psi} \bar{u}, \\
\bar{\Omega} &= \Omega + \E[(a(u) - \bar{a}) (a(u) - \bar{a})^*] - \bar{\Psi} \Pi \bar{\Psi}^*.
\end{align}
\end{subequations}

\subsection{Cubature implementation of statistical linear regression}

The expectations in \eqref{eq:slrexact} are generally intractable.
They can be transformed to integration with respect to a standard Gaussian by change of variables
\begin{equation}\label{eq:change-of-variables}
z = \Pi^{-1/2}(u - \bar{u}).
\end{equation}
Thus for an arbitrary function $\varphi$, its expectation is given by
\begin{equation*}
\E[\varphi(u)] = \E[\varphi(\bar{u} + \Pi^{1/2}z)].
\end{equation*}
One way of approximating these expectations is thus to take
a cubature rule with respect to the standard Gaussian, using $p$ nodes, with weights and nodes given by
\begin{subequations}
\begin{align*}
w &= \begin{pmatrix} w_1 & \ldots & w_p \end{pmatrix}, \\
Z &= \begin{pmatrix} z_1, \ldots, z_p \end{pmatrix}.
\end{align*}
\end{subequations}
This is then transformed into a cubature rule, $(w, U)$, with respect $\mathcal{N}(u; \bar{u}, \Pi)$ via \eqref{eq:change-of-variables}
according to
\begin{subequations}
\begin{align}
\Delta u_i &= \Pi^{1/2} z_i, \\
u_i &= \bar{u} + \Delta u_i,
\end{align}
\end{subequations}
where the $i$th column of $U$ us $u_i$.
The expectation of an arbitrary function $\varphi$, is then approximated by
\begin{equation*}
\E[\varphi(u)] \approx \sum_{i=1}^p w_i \varphi(u_i).
\end{equation*}
Now, define the following quantities
\begin{subequations}
\begin{align}
W &= \diag w \\
\Delta U &= \begin{pmatrix} \Delta u_1, \ldots, \Delta u_p \end{pmatrix}, \\
A &= \begin{pmatrix} a(u_1), \ldots, a(u_p) \end{pmatrix}, \\
\Delta A &= \begin{pmatrix} a(u_1) - \bar{a}, \ldots, a(u_p) - \bar{a} \end{pmatrix}.
\end{align}
\end{subequations}
An approximate solution to the statistical linear regression problem is then given by
\begin{subequations}\label{eq:slrapprox}
\begin{align}
\bar{a} &= A w, \\
\bar{\Psi} &= \Delta A W \Delta U^* \Pi^{-1}, \\
\bar{b} &= \bar{a} - \bar{\Psi} \bar{u},
\end{align}
\end{subequations}
and the residual covariance matrix is approximated by
\begin{equation}\label{eq:approx-residual-covariance}
\bar{\Omega} = \Omega + \Delta A W \Delta A^* - \bar{\Psi} \Pi \bar{\Psi}^*.
\end{equation}

\subsection{Square root implementation of statistical linear regression}
In this section, a downdate free method for implementing the statistical linear regression algorithm is developed.
The following assumption on the chosen cubature rule is required.

\begin{assumption}\label{ass:goodature}
The cubature rule with nodes $\{u_i\}_{i=1}^p$ and weights $\{w_i\}_{i=1}^p$ for $\mathcal{N}(\mu, \Pi)$ satisfy the following:
\begin{itemize}
\item[(i)] It is polynomially exact of at least degree 2.
\item[(ii)] The weights are positive.
\end{itemize}
\end{assumption}

Assumption \ref{ass:goodature} is satisfied by for example Gauss--Hermite quadrautre \cite{Ito2000} or spherical radial cubature \cite{Arasaratnam2009}.
It does not in general hold for the unscented transform \cite{Julier2004}, since for some parameters it has a negative weight \cite{Sarkka2023}.

\begin{proposition}\label{prop:update}
Let $\{x_i\}_{i=1}^p$ and $\{w_i\}_{i=1}^p$ be a cubature rule for $\mathcal{N}(\mu, \Pi)$ that satisfy assumption \ref{ass:goodature}.
Define the matrix:
\begin{equation*}
E = \Delta A - \bar{\Psi} \Delta U,
\end{equation*}
then $\bar{\Omega}$ is given by
\begin{equation}\label{eq:sqrtslr}
\bar{\Omega} = \Omega + E W E^*.
\end{equation}
\end{proposition}

\begin{proof}
By direct calculation from the definition of $E$
\begin{equation*}
\begin{split}
E W E^* &= (\Delta A - \bar{\Psi} \Delta U) W (\Delta A - \bar{\Psi} \Delta U)^* \\
&= \Delta A W \Delta A^* - \Delta A W \Delta U^* \bar{\Psi}^*\\
&\quad - \bar{\Psi} \Delta U W \Delta A^* + \bar{\Psi} \Delta U W \Delta U^* \bar{\Psi}^*
\end{split}
\end{equation*}
and by \eqref{eq:slrapprox} $\Delta A W \Delta U^* = \bar{\Psi} \Pi$, hence
\begin{equation*}
E W E^*
= \Delta A W \Delta A^* - 2\bar{\Psi} \Pi \bar{\Psi}^* + \bar{\Psi} \Delta U W \Delta U^* \bar{\Psi}^*.
\end{equation*}
Furthermore, by assumption \ref{ass:goodature} (polynomial exactness of at least degree 2) $\Delta U W \Delta U^* = \Pi$, therefore
\begin{equation*}
E W E^* = \Delta A W \Delta A^* - \bar{\Psi} \Pi \bar{\Psi}^*.
\end{equation*}
In view of \eqref{eq:approx-residual-covariance}, adding $\Omega$ to both sides of the previous equation gives the desired conclusion.
\end{proof}

Proposition \ref{prop:update} suggests that $\bar{\Omega}$ may be computed with only updates.
This is readily verified since assumption \ref{ass:goodature} ensures that $W$ is positive definite,
and consequently $W^{1/2}$ is well-defined.

\begin{corollary}\label{cor:main}
With assumption \ref{ass:goodature} still in effect, the matrix $\bar{\Omega}$ admits the following factorization
\begin{equation}\label{eq:slrsqrt1}
\bar{\Omega}^{*/2} \cong \begin{pmatrix} \Omega^{1/2} & E W^{1/2} \end{pmatrix}^*.
\end{equation}
\end{corollary}
In view of corollary \ref{cor:main}, a Cholesky factor of $\bar{\Omega}$ may be obtained by taking the upper triangular part in the QR decomposition of
the right-hand side of \eqref{eq:slrsqrt1}.

\subsection{Extension to state-dependent noise}

An extension to state-dependent noise, $\Omega = \Omega(u)$, can be obtained provided
there is a numerically safe way to evaluate $\Omega^{1/2}(u)$ \cite{Yaghoobi2022}.
In this case, the cubature approximation of the residual covariance matrix is given by (c.f. \cite{Tronarp2018})
\begin{equation*}
\bar{\Omega} = \sum_{i=1}^p w_i \Omega(u_i) + \sum_{i=1}^p \Delta A W \Delta A^* - \bar{\Psi} \Pi \bar{\Psi}^*,
\end{equation*}
which by proposition \ref{prop:update} simplifies to
\begin{equation}
\bar{\Omega} = \sum_{i=1}^p w_i \Omega(u_i) + E W E^*.
\end{equation}
The Cholesky factor of $\bar{\Omega}^{*/2}$ may then be obtained by taking the upper triangular factor in the QR decomposition of
\begin{equation}
\begin{pmatrix} E W^{1/2} & \sqrt{w_1}\Omega^{1/2}(u_i) & \ldots & \sqrt{w_p}\Omega^{1/2}(u_p) \end{pmatrix}^*.
\end{equation}

\section{Experimental results}

Consider the continuous-time coordinated turn model \cite{BarShalom2004}
\begin{equation}\label{eq:model}
\dif \begin{pmatrix} p(t) \\ \dot{p}(t) \\ \omega(t) \end{pmatrix}
= \begin{pmatrix} \dot{p}(t) \dif t \\ \omega(t)J p(t) \dif t + B_0 \dif w_p(t) \\ \sigma_\omega \dif w_\omega(t) \end{pmatrix},
\end{equation}
where $p$ is the position, $\dot{p}$ is the velocity, and $\omega$ is the turn rate.
The matrices $J$ and $B_0$ are given by
\begin{equation*}
J = \begin{pmatrix} 0 & -1 \\ 1 & 0 \end{pmatrix}, \quad B_0 = \diag(\sigma_x, \sigma_y).
\end{equation*}
and $w_p, w_\omega$ are standard Wiener processes of appropriate dimension.
An approximation may be obtained by holding $\omega$ dependence of the drift fixed during sampling intervals so that the model is linear on $[t_{m-1}, t_m]$ conditionally on $\omega(t_{m-1})$.
The approximate model may then be written as
\begin{equation*}
\dif x(t) = A(\omega(t_{m-1})) x(t) \dif t + B \dif w(t), \quad t \in [t_{m-1}, t_m),
\end{equation*}
where the definition of $A$ and $B$ are clear from \eqref{eq:model}.
This results in the transition following densities
\begin{equation*}
\mathcal{N}\Big( e^{A(\omega(t_{m-1}))(t_m - t_{m-1})}x(t_{m-1}), Q(\omega(t_{m-1}), t_m - t_{m-1})  \Big).
\end{equation*}
Both the matrix exponential and a Cholesky factor of $Q(\omega(t_{m-1}), t_m - t_{m-1})$ may be computed to high precision
by the algorithm of \cite{Stillfjord2023}.
This leads to a variant of the model used by \cite{Arasaratnam2009},
that is more faithful to the continuous time dynamics.
The initial state distribution is Gaussian with parameters
\begin{subequations}
\begin{align*}
\mu_0 &= \begin{pmatrix} 1000.0 & 1000.0 & 300.0 & 0.0 & -0.0523 \end{pmatrix}, \\
\Sigma_0 &= \operatorname{diag}\begin{pmatrix} 10.0 & 10.0 & 3.162 & 3.162 & 0.316 \end{pmatrix},
\end{align*}
\end{subequations}
the sampling interval is fixed, $\delta t = t_m - t_{m-1}$, and the dynamics parameters are set to
\begin{subequations}
\begin{align*}
\delta t = 1, \quad
\sigma_x = \sigma_y = 0.03, \quad
\sigma_\omega = 0.013.
\end{align*}
\end{subequations}
The state is measured according to
\begin{equation*}
\begin{pmatrix} r(t_m) \\ \theta(t_m) \end{pmatrix}
= \begin{pmatrix} \norm{p(t_m)} + w_r(t_m) \\ \tan^{-1}(p_2(t_m), p_1(t_m)) + w_\theta(t_m) \end{pmatrix},
\end{equation*}
where $w_r(t_m)$ and $w_\theta(t_m)$ are zero mean Gaussian distributed with standard deviations $\sigma_r$ and $\sigma_\theta$, respectively,
given by
\begin{subequations}
\begin{align*}
\sigma_r = 10, \quad
\sigma_\theta = 0.0031.
\end{align*}
\end{subequations}

The proposed implementation (prop) is compared against a reference implementation \cite{Yaghoobi2022} (ref) by
solving the smoothing problem using and iterative statistical linear regression based smoother with 10 iterations \cite{Garcia2016,Tronarp2018}.
The comparison is made in both single and double precision floating point arithmetic.
The norm of the errors in $p(t_m)$, $\dot{p}(t_m)$, and $\omega(t_m)$ are recorded over 100 Monte Carlo generated trajectories of length 101.

The results are shown in figure \ref{fig:results} for the proposed method in both single and double precision arithmetic.
The reference implementation is only shown for double precision as it fails in single precision, due to Cholesky downdates.
As can be seen from the figure, the errors are indistinguishable.
This suggests the proposed implementation ought to be preferred as it more robust to ill-conditioning / low precision arithmetic.

\begin{figure}
\centering
\includegraphics[scale=0.95]{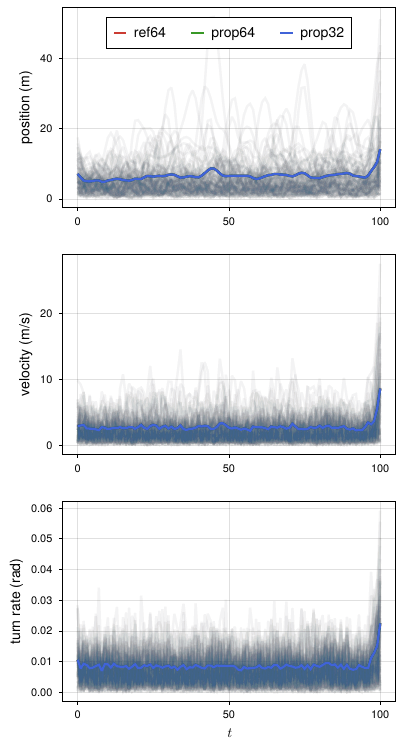}
\caption{
The errors in position (top), velocity (middle), and turn rate (bottom),
for the reference implementation \cite{Yaghoobi2022} (ref) and the proposed implementation (prop).
The average errors across Monte Carlo trials is shown in solid, whereas the error for each trial is shown in transparent color.
The floating point precision used by the algorithm is indicated in the legend by 32 and 64 for single double precision, respectively.
Note that the results of the algorithms are completely indistinguishable.
Except for ref32 which can not run at all because of failing Cholesky downdates,
hence its absence in the figure.
}\label{fig:results}
\end{figure}

\section{Conclusion}
In this article, we have shown that statistical linear regression in square root form
can be implemented in a numerically robust way and demonstrated
its positive impact on tracking performance in an ill-conditioned problem,
where the present best implementation by \cite{Yaghoobi2022} fails in single precision arithmetic.

\bibliographystyle{IEEEtran.bst}
\bibliography{refs}

\end{document}